\documentclass[onecolumn, a4]{IEEEtran}
\IEEEoverridecommandlockouts
\usepackage{cite}
\usepackage{amsmath,amssymb,amsfonts,mathtools, mathrsfs, lipsum, nccmath}
\usepackage{algorithmic}
\usepackage{graphicx}
\usepackage{textcomp}
\usepackage{multicol}

\usepackage[dvipsnames]{xcolor}
\usepackage{cleveref}
\usepackage{etoolbox}
\usepackage{hyperref}

\usepackage{amsthm}  

\theoremstyle{plain} 
\newtheorem{Theorem}{Theorem}
\newtheorem{Lemma}{Lemma}

\newtheorem{Remark}{Remark}

\theoremstyle{TheoremNoColon} 

\newtheoremstyle{break}
  {}
  {}
  {\normalfont}
  {}
  {\bfseries}
  {. }
  {}
  {}

\theoremstyle{break}

\usepackage{tikz}
\usepackage{bbm}

\allowdisplaybreaks

\usepackage{cuted}
\usepackage{xcolor}
\def\BibTeX{{\rm B\kern-.05em{\sc i\kern-.025em b}\kern-.08em
    T\kern-.1667em\lower.7ex\hbox{E}\kern-.125emX}}
\begin{document}

\title{Rate-Distortion Function of the Stochastic Block Model\\
\thanks{This research was funded in whole or in part by EPSRC (EP/T02612X/1) and U.S. Army Research Laboratory and the U.S. Army Research Office (W911NF-22-1-0070). For the purpose of Open Access, the author has applied a CC BY public copyright licence to any Author Accepted Manuscript (AAM) version arising from this submission.}
}

\author{\IEEEauthorblockN{Martin Wachiye Wafula\IEEEauthorrefmark{1}, Praneeth Kumar Vippathalla\IEEEauthorrefmark{2},
Justin Coon\IEEEauthorrefmark{3} and Mihai-Alin Badiu\IEEEauthorrefmark{4}}

\IEEEauthorblockA{Dept. of Engineering Science,
University of Oxford\\
Oxford, United Kingdom\\
Email: \IEEEauthorrefmark{1}martin.wafula,
\IEEEauthorrefmark{2}praneeth.vippathalla,
\IEEEauthorrefmark{3}justin.coon,
\IEEEauthorrefmark{4}mihai.badiu@eng.ox.ac.uk}}

\maketitle
\thispagestyle{plain}
\pagestyle{plain}

\begin{abstract}

The stochastic block model (SBM) is extensively used to model networks in which users belong to certain communities. In recent years, the study of information-theoretic compression of such networks has gained attention, with works primarily focusing on lossless compression. In this work, we address the lossy compression of SBM graphs by characterizing the rate-distortion function under a Hamming distortion constraint. Specifically, we derive the conditional rate-distortion function of the SBM with community membership as side information. We approach this problem as the classical Wyner-Ziv lossy problem by minimising mutual information of the graph and its reconstruction conditioned on community labels. Lastly, we also derive the rate-distortion function of the Erd\H{o}s-R\'enyi (ER) random graph model.

\end{abstract}


\section{Introduction}

Over the past two decades, there has been a significant interest in analysing large graphs representing relationships between various data entities, including those in machine learning, medicine, communications, social and information network analysis, and transportation networks \cite{Eisenman2016ParallelChallenges, Batarfi2015LargeEvaluation, Lakshmanan2016SocialObvious, Besta2018SurveyRepresentations}. Compressing these large graphs has become a topic of interest for academia and industry due to its potential benefits, such as reducing the number of input-output operations, increasing the speed of graph analysis, reducing the amount of data communicated over the network \cite{Besta2019SlimAnalytics} and for efficient storage and transmission of graphs themselves \cite{Badiu2023StructuralGraphs}.

In this work, we examine the stochastic block model (SBM), which partitions the nodes into communities and assigns independent edge probabilities based on the community membership of the nodes. 
The SBM and its variants have been extensively studied and reviewed, particularly in community recovery problems, \cite{Holland1983StochasticSteps, Fortunato2016CommunityGuide, Lee2019AClustering, Saad2018CommunityModel, Abbe2017CommunityDevelopments, Zhao2021OnInformation, Sima2021ExactInformation}. The SBM can be utilised to analyse the community structure of Wireless Sensor Networks (WSNs) as it can capture the natural clustering of nodes based on spatial proximity or functionality, allowing for the investigation of the performance of routing protocols, energy management schemes and other network services that rely on the community structure of the network. For instance, by utilising the community structure to minimise the number of nodes needed to transmit data, energy-efficient routing protocols can be designed for WSNs. Clustered networks outperform non-clustered ones \cite{Jung2023SatelliteApplications}. 

Motivated by their importance in practical applications, many works studied the information-theoretic compression of these graphs. In \cite{Abbe2016GraphClusters}, Abbe addressed the problem of lossless compression of graphs with clusters, while Bhatt et al. dealt with universal compression of SBM \cite{Bhatt2021UniversalModels}. Recently, Han et al. obtained the partitioned structural entropy of SBM similar to Choi's and Szpankowski's \cite{Choi2012CompressionExperiments} structural entropy of the ER graph,  which generalises the structural entropy of unlabelled graphs and encodes the partition information \cite{Han2022StructuralModels}. To consider data on SBM graphs, Asadi, Abbe, and Verd\'u \cite{Asadi2017CompressingClusters} explored data compression limits on graphs using community dependencies, providing optimal compressor lengths when the community signal is strong. They termed the source model data block model.  


There has also been research on lossy graph compression, primarily focusing on heuristics and algorithms with little emphasis on information-theoretic approaches. Preserving graph structure with fidelity is a significant research area \cite{Bustin2017OnMatrices}. Recent work explores compressing directed graphs by maintaining local structure \cite{Bustin2022OnGraphs}, and there is a quantitative definition for compressibility based on network properties like transitivity and degree heterogeneity \cite{Lynn2021QuantifyingNetworks}. SLIMgraph is a recent lossy compression algorithm that has gained attention for its promising results on experimental data \cite{Besta2019SlimAnalytics}. The technique leverages compression kernels and adopts a triangle reduction method, which involves the removal of an edge of the triangle with a specific probability. This removal is done while ensuring that the same edge is not selected for removal again and that edges belonging to multiple triangles, and hence more likely to be part of multiple shortest paths, are not removed more often than those belonging to only one triangle. Another example is shrinkage approximation, a graph summarization algorithm commonly used in call graphs that prunes edges, not nodes \cite{Henecka2015LossyGraphs}. Both SLIMGraph and shrinkage approximation algorithms raise the question of determining the amount of distortion that can be tolerated. Specifically, in these cases, what is the maximum number of edges that can be pruned or discarded during lossy compression while still achieving acceptable performance levels? So far, to the best of our knowledge, there is no work on lossy compression of the SBM. Our work is a step towards addressing this problem. 



Our study uses rate-distortion theory to determine limits for lossy compression algorithms, specifically focusing on the stochastic block model (SBM). We derive a conditional rate with community labels as side information by leveraging the Wyner-Ziv problem\cite{Wyner1976TheDecoder, Yamamoto1982WynerCorresp.}. We extend the findings to both homogeneous and inhomogeneous ER graph models. We apply a simple Hamming distortion measure to the graph edges, which makes practical sense for worst-case edge removal in algorithms such as SLIMGraph. The rate-distortion function serves as a boundary for compression algorithms, utilizing expected distortion as a metric to identify the maximum number of removable edges.


\section{Preliminaries}
In this paper, we restrict our attention to an undirected and unweighted graph $ G = (\mathcal{V}, \mathcal{E})$  constructed on a fixed vertex set $\mathcal{V} = [n]:= \{ 1, 2, \dots, n\}$ where $\mathcal{E} \subseteq \mathcal{V} \times \mathcal{V}$ is the set of edges. Moreover, we assume that there are no self-loops, i.e., $(v,v) \notin \mathcal{E}$ for all $v \in \mathcal{V}$.  It is sometimes useful to represent a graph in terms of its adjacency matrix representation. For a graph $G_n$, we set $E_{i, j} =1$ if the vertices $i$ and $j$ share an edge, and $E_{i, j} =0$ otherwise. As there are no self-loops, which means that $E_{ii}=0$ for $i \in [n]$, the values $\{E_{i, j}: 1\leq i<j \leq n\}$ specify the adjacency matrix uniquely, and hence the graph. Next, we present a few fundamental results related to the entropy of various random graph models. 

\subsection{Erd\H{o}s-R\'enyi Graph Model} 
 In the Erd\H{o}s-R\'enyi (random) graph model $\mathcal{G}(n,p)$ model, we connect any two distinct vertices in $[n]$ by an edge with probability $p$ independently of the other edges. The ER model is the canonical example, to begin with, in most studies on graphs. One observes that under this model, the probability of a graph $G_n$ containing $k$ edges is $P(G) = p^k(1-p)^{\binom{n}{2} - k}.$

 In a generalisation of the ER model, which is commonly referred to as an inhomogeneous ER graph model or a generalised binomial graph,  the probability of connecting two vertices $i$ and $j$ with $i<j$ by an edge is $p_{i, j}$. 
 This model was first defined by Kovalenko in \cite{Kovalenko1971TheoryGraphs}. Under this model, the probability of a graph $G_n$ specified by $E_{i, j} \in \{0,1\}$ for $i<j$ is given  by
 $$P(G) = \prod_{i<j}p_{i, j}^{E_{i, j}}(1-p_{i, j})^{1 - E_{i, j}}$$.
\begin{Lemma}\label{EntropyERGeneral}
The entropy of an inhomogeneous ER graph model with edge probabilities $\{p_{i, j}: 1 \leq i <j \leq n\}$ is \begin{equation*}
    H(G) = \sum_{i<j} h_2(p_{i, j}).
\end{equation*} 
In particular, the entropy of the ER graph model $\mathcal{G}(n,p)$ is $ H(G) = {\binom{n}{2}}h_2(p).$
Here $h_2(t) := - t\log_2 t - (1-t)\log_2 (1-t)$ is the binary entropy function.
\end{Lemma}

\subsection{Stochastic Block Model (SBM)}

{\color{black}

Let $\mathbf{X} = (X_1, \dotsc, X_n) \in [k]^n$ be a random vector whose components are identically and independently distributed (i.i.d.) according to $\mathbf{p} = (p_1, \dotsc, p_k)$. The components describe the community membership of $n$ nodes. Let $\mathbf{W}$ be a $k \times k$ symmetric matrix where the entry $w_{\ell, m}$ specifies the probability that a node in community  $\ell$ is connected to another node in community $m$. In the  \emph{stochastic block model}, denoted by $(\mathbf{X}, G_n) \sim {\text{SBM}(n, \mathbf{p},\mathbf{W})}$, we generate  a random graph $G_n$ on $n$ nodes by drawing an edge between nodes $i$ and $j$ with probability $w_{X_i, X_j}$ (using the community membership vector $\mathbf{X}$). Notice that given a community label vector $\mathbf{X}$, the graph $G_n$ (or equivalently, $\{E_{i, j}: 1\leq i<j \leq n\}$) is generated according to the inhomogeneous ER model with edge probabilities  $\{w_{X_i,X_j}: 1 \leq i <j \leq n\}$.

   

\begin{Lemma}[Abbe, 2016\cite{Abbe2016GraphClusters}] \label{EntropySBM} Let $(\mathbf{X}, G_n) \sim{\text{SBM}(n, \mathbf{p},\mathbf{W})}$. The entropy of the SBM graph is 
\begin{equation}
    H(G_n) = \binom{n}{2}\mathbf{p}^Th_2(\mathbf{W})\mathbf{p} + \mathcal{O}(n),\label{EntropySBMeqn}
\end{equation}
where $h_2(\mathbf{W})$ is an $n \times n$ matrix whose $(i,j)$-th entry is $h_2(w_{i, j})$.
The entropy of the SBM graph conditioned on the community labels is given by $H(G_n|\mathbf{X}) =  \binom{n}{2}\mathbf{p}^Th_2(\mathbf{W})\mathbf{p}.$
\end{Lemma}
}

 \subsection{Rate-Distortion Theory}
In this section, we will review rate-distortion theory, which is primarily used to study lossy compression. Consider two measurable spaces $(\mathcal{U},\mathcal{B_U})$ and $(\hat{\mathcal{U}},\mathcal{B_{\hat{U}}})$ for the source and its compressed version, respectively, equipped with a nonnegative measurable mapping $d \colon \mathcal{U \times \hat{U}} \to \mathbb{R}\cup \{+\infty\}$. This mapping is referred to as the distortion function. Suppose a distortion level $D$ and probability measure $P$  on $(\mathcal{U},\mathcal{B_U})$ are given. In lossy compression, one wishes to represent the elements $u$ of $\mathcal{U}$ using the elements $\hat{u}$ of the $\hat{\mathcal{U}}$ such that the average of the distortions $d(u,\hat{u})$ is at most $D$. The aim is to achieve this in a parsimonious way, for example, by using as few elements of $\hat{\mathcal{U}}$ as possible or by associating a fixed or variable-length codeword to each element of $\hat{\mathcal{U}}$ and then minimizing the average length of the codewords corresponding to a (lossy) representation. 

The focus of the classical rate-distortion theory is sequential (random) sources $X^n$ taking values in $\mathcal{X}^{n}$ using the sequences from   $\hat{\mathcal{X}}^{n}$ for $n \geq 1$. Given a single-letter distortion measure $d': \mathcal{X} \times \hat{\mathcal{X}} \to \mathbb{R}$,  the distortion measure for sequences is defined as $d(x^n, {\hat{x}}^n) =  \sum_{i=1}^{n}d'(x_i, \hat{x}_{i})/n$.
We say that a rate $R$ is achievable at a distortion level $D$ if there exists a sequence of codes $g_n \colon \mathcal{X}^{n} \to \hat{\mathcal{X}}^{n}$ such that $\limsup_{n \to \infty} \log|g_n(\mathcal{X}^{n})|/n \leq R$ and $\limsup_{n \to \infty}\mathbb{E}[d( X^n, g(X^{n}))] \leq D$. The goal is to characterize the smallest achievable rate, denoted by $R_X(D)$, at a distortion level $D$. In his seminal work \cite{Shannon1959CodingConv}, Shannon showed that for an i.i.d. source $X^n$,
\begin{equation}\label{eq:rdf:shannon}
        R_X(D) = \min_{p_{\hat{X}|X:\mathbbm{E}[d(X, \hat{X})] \leq D}} I(X; \hat{X}).
\end{equation}

The rate-distortion function $R(D)$ is a non-increasing convex function of $D \in [0,\infty)$ \cite{Cover2006ElementsProcessing}. In \cite{Wyner1976TheDecoder}, Wyner and Ziv studied the lossy compression in the presence of side information $Y^n \in \mathcal{Y}^{n}$, which is correlated with the source $X^n$. In this case, the map $g_n$ is a function of $\mathcal{X}^{n} \times \mathcal{Y}^{n}$, and $R_{X|Y}(D)$ is the smallest achievable rate at a distortion level $D$. It was shown in \cite{Wyner1976TheDecoder} that if $(X^n,Y^n)$ is jointly i.i.d., then
\begin{equation}\label{eq:rdf:wyner}
        R_{X|Y}(D) = \min_{p(\hat{x}| x, y): \mathbbm{E}[d(X, \hat{X})] \leq D} I (X; \hat{X}| Y). 
\end{equation}
Interestingly, the rate-distortion functions \eqref{eq:rdf:shannon} and \eqref{eq:rdf:wyner} satisfy \cite{Gray1973AFunctions} the inequality
\begin{equation}\label{eq:gray:cond_rate}
        R_{X|Y}(D) \leq R_{X}(D) \leq R_{X|Y}(D) + I(X; Y),
\end{equation}
which is analogous to the entropic inequality $H(X|Y) \leq H(X) = H(X|Y) + I(X; Y)$.

In general, sources of interest need not have a sequential structure; for example, graphical sources. In such cases, we can still study the mutual information between a source and its reconstructed version, minimized over all conditional distributions satisfying the distortion condition, as in \eqref{eq:rdf:shannon}. It can be argued that this quantity gives a converse bound on the compression rate (appropriately scaled) in a lossy compression problem. 

\section{Rate Distortion Function of the SBM}


We focus on the Hamming distortion measure on the set of all graphs with $n$ vertices. For two graphs $G_n$ and $\hat{G}_n$, the Hamming distortion is defined as $d(G_n, \hat{G}_n)= \sum_{i<j}\mathbbm{1}(E_{i,j} \neq \hat{E}_{i,j}),$
where the edges $E_{i,j}$ and $\hat{E}_{i,j}$, respectively, correspond to the graphs  $G_n$ and $\hat{G}_n$, and $\mathbbm{1}(.)$ is the indicator function. Clearly, $0\leq d(G_n, \hat{G}_n) \leq \binom{n}{2}$. For a stochastic block model $(\mathbf{X}, G_n)\sim \text{SBM}(n,\mathbf{p},\mathbf{W})$, the rate-distortion function at the Hamming distortion $D$ is defined as
\begin{align}\label{eq:sbm:no:conditioning}
    R_{G_n}(D) \triangleq \min_{P_{\hat{G}_n|G_n}: \mathbb{E}[d(G_n,\hat{G}_n)] \leq D } I(G_n;\hat{G}_n),
\end{align}
where $\hat{G}_n$ represents the reconstructed version of the original graph $G_n$. Notice that if $\binom{n}{2} \min\{\sum_{l,m} p_l p_m  w_{l,m}, \linebreak \sum_{l,m} p_l p_m \left(1-w_{l,m}\right)\} \leq D$, $R_{G_n}(D)=0$ because $\hat{G}_n$ can be made independent of $G_n$ while satisfying the constraint\footnote{We use the inequality $\mathbb{E}[d(G_n,\hat{G}_n)] 
= \sum_{i<j} \mathbb{P}(E_{i, j}\oplus \hat{E}_{i, j}=1)\geq \sum_{i<j} \min\{\mathbb{P}(E_{ij}=0), \mathbb{P}({E}_{ij}=1)\}$, which holds when $G_n$ and $\hat{G}_n$ are independent.} $\mathbb{E}[d(G_n,\hat{G}_n)] \leq D$. So in this case, we restrict to the case $0\leq D\leq \binom{n}{2} \min\{\sum_{l,m} p_l p_m  w_{l,m}, \sum_{l,m} p_l p_m \left(1-w_{l,m}\right)\}$.

In the case when  $\mathbf{X}$ is available as side information to both the encoder and decoder, the rate-distortion function (with side information $\mathbf{X}$) at the Hamming distortion $D$ is defined as
\begin{align}
    R_{G_n|\mathbf{X}}(D) \triangleq \min_{P_{\hat{G}_n|G_n, \mathbf{X}}: \mathbb{E}[d(G_n,\hat{G}_n)] \leq D } I(G_n;\hat{G}_n|  \mathbf{X}).
\end{align}
As in previous case, if $\binom{n}{2} \sum_{l,m} p_l p_m  \min\{ w_{l,m}, 1-w_{l,m} \} \leq D$, then\footnote{When $G_n$ is independent of $\hat{G}_n$ conditioned on $\mathbf{X}$, it holds that $\mathbb{E}[d(G_n,\hat{G}_n)] 
= \sum_{i<j} \mathbb{P}(E_{i, j}\oplus \hat{E}_{i, j}=1)\geq \sum_{i<j} \sum_{l,m} p_l p_m \min\{\mathbb{P}(E_{ij}=0|X_i = l, X_j = m), \mathbb{P}({E}_{ij}=1|X_i = l, X_j = m)\}$.} $R_{G_n|\mathbf{X}}(D)=0$. Therefore, the interval that is of interest here is $0\leq D\leq \binom{n}{2} \sum_{l,m} p_l p_m  \min\{ w_{l,m}, 1-w_{l,m} \}$.

In the next theorem, we give a characterization of the rate-distortion function of the stochastic block model.

\begin{Theorem}[RDF of SBM]\label{RDFSBMcommunityknown}
Let $(\mathbf{X}, G_n)\sim \text{SBM}(n,\mathbf{p},\mathbf{W})$.
\begin{enumerate} 
    \item  For $0\leq D\leq \binom{n}{2} \sum_{l,m} p_l p_m  \min\{ w_{l,m}, 1-w_{l,m} \}$, the rate-distortion function of the SBM model at the Hamming distortion $D$ is 
    \begin{equation}\label{eq:rdf_sbm:uncond}
        R_{G_n}(D) = \binom{n}{2} \mathbf{p}^T \left[ h_2(\mathbf{W}) - h_2(\mathbf{D}^*)\right] \mathbf{p} + \mathcal{O}(n),
    \end{equation}
    where $\mathbf{D}^*$ is a matrix with entries chosen according to ${d}^*_{l,m} = \min\left\{ \min\left\{ w_{l,m}, 1- w_{l,m}\right\}, \mu  \right\},$  for ${l,m \in [k]},$
and $\mu$ is chosen such that the constraint $ \mathbf{p}^T\mathbf{D}^* \mathbf{p} = \frac{D}{\binom{n}{2}}$.
    The rate-distortion function of the SBM graph conditioned on the community labels is given by 
    \begin{equation}\label{eq:rdf_sbm:condtional}
        R_{G_n|\mathbf{X}}(D) = \binom{n}{2} \mathbf{p}^T \left[ h_2(\mathbf{W}) - h_2(\mathbf{D}^*)\right] \mathbf{p}
    \end{equation}
    with $\mathbf{D}^*$ defined above.
    \item For $\binom{n}{2} \sum_{l,m} p_l p_m  \min\{ w_{l,m}, 1-w_{l,m} \} \leq D \leq \min\{\sum_{l,m} p_l p_m  w_{l,m}, \sum_{l,m} p_l p_m \left(1-w_{l,m}\right)\}$,  $R_{G_n|\mathbf{X}}(D) =0$ and  $R_{G_n}(D) =\mathcal{O}(n)$
    \end{enumerate}
\end{Theorem}

\begin{proof}
    See Section~\ref{sec:proof:sbm}.
\end{proof}

The result in~\eqref{eq:rdf_sbm:condtional} is shown in Figure \ref{fig:RDF_Curve}. 

\begin{Remark}
    The value of $\mu$ in the above solution can be chosen by ``reverse water-filling," but now the levels $d_{l,m}^*$ scaled by $p_lp_m$ should add up to $\frac{D}{\binom{n}{2}}$.
\end{Remark}

\begin{figure}[h]
    \centering
    \includegraphics[scale = 0.35]{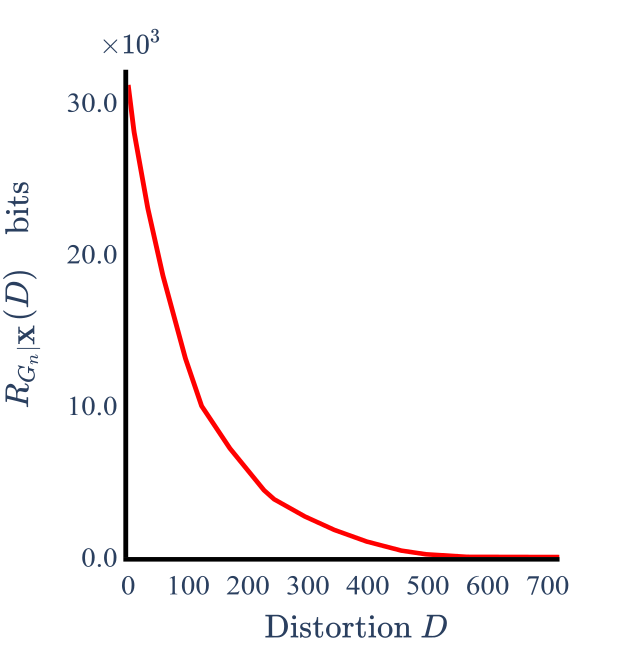}
    \caption{Rate-distortion function \eqref{eq:rdf_sbm:condtional} of an SBM over 100 nodes with $k= 3$, $\mathbf{p}=(0.4, 0.3, 0.3)$ and $\mathbf{W} =  [0.5, 0.2, 0.1; 0.2, 0.5, 0.1; 0.1, 0.1, 0.4]$.}
    \label{fig:RDF_Curve}
\end{figure}

A similar characterization can also be obtained for the rate-distortion function of the (inhomogeneous) ER model, which is defined as in \eqref{eq:sbm:no:conditioning}.


\begin{Theorem}[RDF of the inhomogeneous ER Model]\label{RDFERbinom}
Given an inhomogeneous ER graph with an edge probability, $p_{i, j}$ between any two nodes $i$ and $j$, $i,j \in \mathcal{V}$ and Hamming distortion $0\leq D\leq \sum_{i<j}\min\{p_{i,j}, 1-p_{i,j}\}$, the rate-distortion function of the graph is given by 

\begin{equation}
    R(D) = \sum_{i < j} \left(h_2(p_{i, j}) - h_2(d_{i, j})\right)
\end{equation}
where $d_{i, j} = \min \{ \min\{p_{i, j}, 1- p_{i, j}\}, \lambda \}$ for ${i,j \in \mathcal{V}}$
and $\lambda$ is chosen such that $\sum_{i<j} d_{i, j} = D$. 
\end{Theorem}
\begin{proof}
    See Section~\ref{sec:proof:ER}.
\end{proof}

From the above theorem, we can easily see that the rate-distortion function of a Erd\"os-R\'enyi(ER) graph model $\mathcal{G}(n,p)$ with the Hamming distortion $D$ is given by
\begin{equation}
    R(D) = \begin{cases} 
            \binom{n}{2}\left[h_2(p) - h_2\Big(\frac{D}{\binom{n}{2}}\Big)\right], & \text{if} \quad 0 \leq D \leq  \binom{n}{2} \min\{p,1-p\} \\
            0, & \text{otherwise}.
            \end{cases} \label{RDFERG_np}
    \end{equation}


\section{Proofs of Theorem~\ref{RDFSBMcommunityknown} and Theorem~\ref{RDFERbinom}}

\subsection{Proof of Theorem~\ref{RDFSBMcommunityknown}}\label{sec:proof:sbm}


By applying \eqref{eq:gray:cond_rate} to the pair $( \mathbf{X}, G_n)$, we obtain
\begin{align}
    R_{G_n|\mathbf{X}}(D) \leq R_{G}(D) \leq R_{G_n|\mathbf{X}}(D) + I(G;\mathbf{X}).
\end{align}
Since $I(G;\mathbf{X})\leq H(\mathbf{X}) = n H(\mathbf{p})$, we have $R_{G_n|\mathbf{X}}(D) \leq R_{G}(D) \leq R_{G_n|\mathbf{X}}(D) + \mathcal{O}(n).$ The second part of the theorem immediately follows because  $R_{G_n|\mathbf{X}}(D) =0$ when $\binom{n}{2} \sum_{l,m} p_l p_m  \min\{ w_{l,m}, 1-w_{l,m} \} \leq D \leq \min\{\sum_{l,m} p_l p_m  w_{l,m}, \linebreak \sum_{l,m} p_l p_m \left(1-w_{l,m}\right)\}$. For the first part, it is enough to show that $R_{G_n|\mathbf{X}}(D)$ satisfies \eqref{eq:rdf_sbm:condtional}. 

To this end, consider a conditional joint distribution $P_{\hat{G}_n\mid G_n, \mathbf{X}}$ satisfying the constraint $\mathbb{E}[d(G_n,\hat{G}_n)] \leq D$.  For this distribution, we have 
 \begin{align}
    I(G_n;\hat{G}_n|  \mathbf{X}) &= H(G_n|  \mathbf{X}) -H(G_n| \hat{G}_n,  \mathbf{X})\nonumber\\
    & = \sum_{i<j} H(E_{i, j}| X_i, X_j) -H(G_n| \hat{G}_n,  \mathbf{X}) \label{eq:eij_cond_xi_xj} \\
    & = \sum_{i<j} H(E_{i, j}| X_i, X_j) -  \sum_{i<j} H(E_{i, j}| E_{12}, \dotsc, E_{i-1,j}, \hat{G}_n, \mathbf{X}) \nonumber \\
    & \geq \sum_{i<j} H(E_{i, j}| X_i, X_j) -  \sum_{i<j} H(E_{i, j}| \hat{E}_{i, j}, X_i, X_j) \nonumber \\
    & = \sum_{i<j}  \sum_{l,m} p_l p_m \left[ H(E_{i, j}|X_i = l, X_j = m)  - H(E_{i, j}| \hat{E}_{i, j}, X_i = l, X_j = m) \right]  \nonumber\\
     & = \sum_{i<j}  \sum_{l,m} p_l p_m \left[H(E_{i, j}|X_i = l, X_j = m)  - H(E_{i, j} \oplus \hat{E}_{i, j}| \hat{E}_{i, j}, X_i = l, X_j = m) \right]  \nonumber\\
     & \geq \sum_{i<j}  \sum_{l,m} p_l p_m \max\left\{H(E_{i, j}|X_i = l, X_j = m) - H(E_{i, j} \oplus \hat{E}_{i, j}| X_i = l, X_j = m), \ 0 \right\}\label{eq:cnd_reduction:1}\\
     & = \sum_{i<j}  \sum_{l,m} p_l p_m \max\left\{h_2(w_{l,m}) - h_2(d_{l,m}^{i,j}), \ 0 \right\}\label{eq:define:D_ijs}\\
     & = \sum_{i<j}  \sum_{l,m} p_l p_m \left[h_2(w_{l,m}) - \min\{ h_2(d_{l,m}^{i,j}), h_2(w_{l,m}) \} \right]\nonumber\\
     & = \binom{n}{2} \mathbf{p}^Th_2(\mathbf{W})\mathbf{p} - \sum_{i<j}  \sum_{l,m} p_l p_m \min\{ h_2(d_{l,m}^{i,j}), h_2(w_{l,m}) \}\nonumber\\
      & \geq \binom{n}{2} \mathbf{p}^Th_2(\mathbf{W})\mathbf{p} -   \sum_{l,m} p_l p_m \min\left\{ \sum_{i<j} h_2(d_{l,m}^{i,j}), \binom{n}{2}h_2(w_{l,m}) \right\} \label{Inequality}\\
      & \geq \binom{n}{2} \mathbf{p}^Th_2(\mathbf{W})\mathbf{p} -  \binom{n}{2} \sum_{l,m} p_l p_m \min\left\{ h_2(d_{l,m} ), h_2(w_{l,m}) \right\}.\label{AverageEntropy}
\end{align}
In \eqref{eq:eij_cond_xi_xj}, we use the fact that $P_{\{E_{i,j}:i<j\}\mid \mathbf{X}}=\prod_{(i,j)}P_{E_{i, j}\mid X_i, X_j}$. \eqref{eq:cnd_reduction:1} is from the fact that mutual information is nonnegative. In \eqref{eq:define:D_ijs}, we define $$d_{l,m}^{i,j} \triangleq \min\left\{ \mathbb{P}(E_{i, j} \oplus \hat{E}_{i, j} = 1| X_i = l, X_j = m),\  \mathbb{P}(E_{i, j} \oplus \hat{E}_{i, j} = 0| X_i = l, X_j = m)\right\}$$ for  $l, m \in [k]$ \text{ and } $i,j \in \mathcal{V}$. For \eqref{Inequality}, we rely on a simple fact that for two real number sequences $(a_1, a_2, ..., a_n)$ and  $(b_1, b_2, ..., b_n)$, $ \sum_i \min\left\{ a_i,  b_i\right\} \leq \min \left\{\sum_i a_i, \sum_i b_i\right\}$. The inequality  \eqref{AverageEntropy}, where we set  $d_{l,m}\triangleq \sum_{i<j} \frac{d_{l,m}^{i,j}}{\binom{n}{2}}$, follows immediately from Jensen's inequality and concavity of the entropy : 
$$  \sum_{i<j} h_2(d_{l,m}^{i,j}) =\binom{n}{2} \sum_{i<j} \frac{1}{\binom{n}{2}}h_2(d_{l,m}^{i,j})   \leq 
        \binom{n}{2} h_2\left(\sum_{i<j} \frac{d_{l,m}^{i,j}}{\binom{n}{2}} \right) =  \binom{n}{2} h_2(d_{l,m} ).$$
We can show that $d_{l,m}$'s satisfy the condition that $\sum_{l,m} d_{l,m} p_l p_m \leq \frac{D}{\binom{n}{2}}$ by using the distortion criterion: 
   \begin{align}
       D \geq \mathbb{E}[d(G_n,\hat{G}_n)] = \sum_{i<j} \mathbb{P}(E_{i, j}\oplus \hat{E}_{i, j}=1)&=\sum_{i<j} \sum_{l,m} p_l p_m \mathbb{P}(E_{i, j}\oplus \hat{E}_{i, j}=1|X_i = l, X_j = m)   \nonumber\\
       &\geq \sum_{i<j} \sum_{l,m} p_l p_m d_{l,m}^{i,j}  \nonumber\\ 
       &=  \binom{n}{2} \sum_{l,m} p_l p_m  d_{l,m}  \nonumber
   \end{align}
From \eqref{AverageEntropy}, we have 
\begin{align}
    I(G_n; \hat{G}_n |  \mathbf{X}) \geq \min\left\{\binom{n}{2} \mathbf{p}^Th_2(\mathbf{W})\mathbf{p} -   \binom{n}{2} \sum_{l,m} p_l p_m \min\left\{ h_2(d_{l,m} ), h_2(w_{l,m}) \right\}:\sum_{l,m} d_{l,m} p_l p_m \leq \frac{D}{\binom{n}{2}}\right\},\nonumber 
\end{align}
which implies that the rate-distortion function of the SBM condition on community labels should be 
\begin{align}
    R_{G_n|\mathbf{X}}(D) &\geq \min\left\{\binom{n}{2} \mathbf{p}^Th_2(\mathbf{W})\mathbf{p} -   \binom{n}{2} \sum_{l,m} p_l p_m \min\left\{ h_2(d_{l,m} ), h_2(w_{l,m}) \right\}:\sum_{l,m} p_l p_m d_{l,m}  \leq \frac{D}{\binom{n}{2}}\right\}\nonumber \\
    &=\binom{n}{2} \mathbf{p}^Th_2(\mathbf{W})\mathbf{p} -   \binom{n}{2}\max\left\{  \sum_{l,m} p_l p_m \min\left\{ h_2(d_{l,m} ), h_2(w_{l,m}) \right\}:\sum_{l,m} p_l p_m d_{l,m}  \leq \frac{D}{\binom{n}{2}}\right\}.\label{RDF_SBM_UpperBound1}
\end{align}

Notice that the optimization problem in \eqref{RDF_SBM_UpperBound1} is equivalent to the following optimization problem:
\begin{equation}
\begin{aligned}
& \text{minimize}
& & -\sum_{l,m} p_l p_m  h_2(d_{l,m} ) & \\
& \text{subject to}
& & \sum_{l,m}  p_l p_m d_{l,m} \leq  \frac{D}{\binom{n}{2}}, & {l,m \in [k]}, \\
&&& d_{l,m} \leq \min \{w_{l,m}, 1 - w_{l,m}\}, & l,m \in [k].\label{SBMOptProblemStandard}
\end{aligned}
\end{equation}
This is simply because we can map an optimizer of \eqref{RDF_SBM_UpperBound1} to a point in the constraint of \eqref{SBMOptProblemStandard} such that the objective functions of both these problems evaluated at these points remain the same, and vice versa.
The Lagrangian for the above optimization problem is \begin{equation}
    L = -\sum_{l,m} p_l p_m  h_2(d_{l,m} ) + \sum_{l,m} \lambda_{l,m} \left(d_{l,m} - \min \{w_{l,m}, 1 - w_{l,m}\}\right) +  \nu \left( \sum_{l,m}  p_l p_m d_{l,m} - \frac{D}{\binom{n}{2}}\right). \label{LangrangianSBM}
\end{equation} 
The corresponding KKT conditions are given as follows
 \begin{align}
 &\left.\frac{\partial L}{\partial d_{l,m}}\right|_{d_{l,m}= d^{*}_{l,m}} = - p_l p_m \log\left(\frac{1 - d^{*}_{l,m}}{d^{*}_{l,m}}\right) + \lambda_{l,m} + \nu p_l p_m= 0 , \quad l,m \in [k]\label{DerivativeZeroSBM},\\
       &d^{*}_{l,m} - \min \{w_{l,m}, 1 - w_{l,m}\} \leq 0, \quad l,m \in [k]\label{eq:dlm:lessthan}, \\   
       &\sum_{l,m}  p_l p_m d^{*}_{l,m} - \frac{D}{\binom{n}{2}} \leq 0,  \label{ConstraintSBM},\\
       &\nu \left( \sum_{l,m}  p_l p_m d_{l,m} - \frac{D}{\binom{n}{2}}\right)=0\label{eq:sum:leq},\\
&\lambda_{l,m} (d^{*}_{l,m} - \min \{w_{l,m}, 1 - w_{l,m}\}) = 0, \quad l,m \in [k] \label{ComplementarySlacknessSBM},\\
       &\nu\geq0, \quad \lambda_{l,m} \geq 0, \quad l,m \in [k]  \label{LambdaSBM}.
         \end{align}
         
Since \eqref{SBMOptProblemStandard} is a convex optimization problem function satisfying Slater's condition, the KKT conditions give sufficient and necessary conditions for an optimizer, that is,  $\nu$, $\lambda_{l,m}$ and $d^{*}_{l,m}$ for $l,m \in [k]$ satisfy the KKT conditions if and only if $(d^{*}_{l,m})_{l,m \in [k]}$ is an optimizer of \eqref{SBMOptProblemStandard}. To find a solution to the KKT conditions, we can consider the equation \eqref{DerivativeZeroSBM}, which can be written as 
\begin{equation}
   d^{*}_{l,m} = \frac{1}{\exp{\left(\frac{\lambda_{l,m}}{p_lp_m}+ \nu\right)} + 1}. \label{D_lmSBM}
\end{equation}
for $l,m \in [k]$.
Let us consider two cases for each $l,m \in [k]$.
\begin{itemize}
    \item \textbf{Case 1}: When $\lambda_{l,m} = 0$, we have $d^{*}_{l,m} = \frac{1}{e^{ \nu} + 1}$, which is upper bounded by $\min \{w_{l,m}, 1 - w_{l,m}\}$ as a result of \eqref{eq:dlm:lessthan}.
    \item \textbf{Case 2}: When $\lambda_{l,m} > 0$, we have $d^{*}_{l,m} = \min \{w_{l,m}, 1 - w_{l,m}\} $ by \eqref{ComplementarySlacknessSBM}. This together with \eqref{D_lmSBM} yields $ d^{*}_{l,m} = \min \{w_{l,m}, 1 - w_{l,m}\}= \frac{1}{\exp{\left(\frac{\lambda_{l,m}}{p_lp_m}+ \nu\right)} + 1}$, which is upper bounded by $\frac{1}{e^{ \nu} + 1}$ for $\lambda_{l,m} > 0$.
\end{itemize}
 Therefore, by combining these two cases and setting $\mu\triangleq \frac{1}{e^{ \nu} + 1}$, we see that the optimizer of \eqref{SBMOptProblemStandard} is of the form \begin{equation*}
        d_{l,m}^* = \min\left\{\min\{w_{l,m}, 1 - w_{l,m}\}, \mu \right\} \end{equation*} 
   for $l,m \in [k]$   
        where $\mu$ is chosen such that $\sum_{l,m}  p_l p_m d_{l,m}^* \leq  \frac{D}{\binom{n}{2}}$, and $\lambda_{l,m}\geq 0$ are chosen according to \eqref{ComplementarySlacknessSBM}. Now, let us argue that we cannot have $\sum_{l,m}  p_l p_m d_{l,m}^* <  \frac{D}{\binom{n}{2}}$. If it were the case, then $\nu=0$ by \eqref{eq:sum:leq}, implying that $\mu=\frac{1}{2}$ and $d_{l,m}^* = \min\{w_{l,m}, 1 - w_{l,m}\}$. This contradicts the fact that $D\leq \binom{n}{2} \sum_{l,m} p_l p_m  \min\{ w_{l,m}, 1-w_{l,m} \} $. Therefore, the $d^{*}_{l,m}$'s satisfy the condition $\sum_{l,m}  p_l p_m d_{l,m}^* =  \frac{D}{\binom{n}{2}}$.

        Hence, we obtain from \eqref{RDF_SBM_UpperBound1} that 
\begin{align}
    R_{G_n|\mathbf{X}}(D) \geq \binom{n}{2} \mathbf{p}^Th_2(\mathbf{W})\mathbf{p} - \binom{n}{2} \sum_{l,m} p_l p_m h_2(d_{l,m}^*) = \binom{n}{2} \mathbf{p}^T \left[ h_2(\mathbf{W}) - h_2(\mathbf{D}^*)\right]\mathbf{p}, \label{RDF_SBM_lowerBound}
\end{align}
    where $\mathbf{D}^*$ is a matrix with entries chosen according to ${d}^*_{l,m} = \min\left\{ \min\left\{ w_{l,m}, 1- w_{l,m}\right\}, \mu  \right\}, l,m \in [k],$
and $\mu$ is chosen such that the constraint $ \mathbf{p}^T\mathbf{D}^* \mathbf{p} = \frac{D}{\binom{n}{2}}$. 

{\textbf{Achievability of the lower bound:}}
Now, it remains to show that inequality in \eqref{RDF_SBM_lowerBound} holds with equality. To this end, we have to produce a conditional probability distribution $P_{\hat{G}_n|G_n,\mathbf{X}}$ (or equivalently, $P_{\{\hat{E}_{i,j}:{i<j}\}|\{E_{i,j}:{i<j}\},\mathbf{X}}$) that satisfies the distortion constraint and has $I(G_n; \hat{G}_n |  \mathbf{X})$ equal to the expression on the right-hand side of \eqref{RDF_SBM_lowerBound}. Let $d_{l,m}^*$ be chosen according to the constraint in \eqref{RDF_SBM_lowerBound}. Consider the conditional probability distribution of the form
\begin{align}
    P_{\{\hat{E}_{i,j}:{i<j}\}|\{E_{i,j}:{i<j}\},\mathbf{X}}= \prod_{i<j} P_{\hat{E}_{i, j}|E_{i, j},X_i, X_j}=\prod_{i<j} \frac{P_{{E}_{i, j}|\hat{E}_{i, j}, X_i, X_j}P_{\hat{E}_{i, j}| X_i, X_j}}{P_{{E}_{i, j}|X_i, X_j}},
\end{align}
 where for $l,m \in [k]$, $P_{{E}_{i, j}|\hat{E}_{i, j}, X_i, X_j}(\cdot \mid \cdot,l,m)$ is a binary symmetric channel with the crossover probability $d^*_{l,m}$ and $P_{\hat{E}_{i, j}| X_i, X_j}(1\mid l,m)=\frac{w_{l,m} -d_{l,m}^*}{1 - 2d_{l,m}^*}$ and $P_{\hat{E}_{i, j}| X_i, X_j}(0\mid l,m)=\frac{1-w_{l,m} -d_{l,m}^*}{1 - 2d_{l,m}^*}$, which are non-negative by virtue of the way $d^*_{l,m}$'s are defined. (Note that these distributions are the same for all $i<j$.) It is easy to see that this distribution satisfies the distortion criterion:
 \begin{align}
       \mathbb{E}[d(G_n,\hat{G}_n)] = \sum_{i<j} \mathbb{P}(E_{i, j}\oplus \hat{E}_{i, j}=1)&=\sum_{i<j} \sum_{l,m} \mathbb{P}(E_{i, j}\oplus \hat{E}_{i, j}=1|X_i = l, X_j = m) p_l p_m  \nonumber\\&= \sum_{i<j} \sum_{l,m} d^*_{l,m} p_l p_m \nonumber\\&=  \binom{n}{2} \mathbf{p}^T\mathbf{D}^* \mathbf{p} \nonumber \\&= D  \nonumber.
   \end{align}
The conditional mutual information is 
\begin{align}
    I(G_n; \hat{G}_n |  \mathbf{X}) = H(G_n|  \mathbf{X}) - H(G_n| \hat{G}_n, \mathbf{X})  &= \sum_{i<j} \sum_{l,m} p_l p_m h_2(w_{l,m})  - \sum_{i<j} \sum_{l,m} p_l p_m h_2(d_{l,m}^*)\nonumber \\
    &=  \binom{n}{2} \mathbf{p}^T h_2(\mathbf{W}) \mathbf{p} - \binom{n}{2} \mathbf{p}^T h_2(\mathbf{D}^*) \mathbf{p},    
\end{align}
which implies that 
\begin{align}
     R_{G_n|\mathbf{X}}(D) \leq \binom{n}{2} \mathbf{p}^T h_2(\mathbf{W}) \mathbf{p} - \binom{n}{2} \mathbf{p}^T h_2(\mathbf{D}^*) \mathbf{p} \label{RDF_SBM_UpperBound}.
\end{align}
By combining \eqref{RDF_SBM_UpperBound} and \eqref{RDF_SBM_lowerBound}, we get the desired result.

\subsection{ Proof of Theorem~\ref{RDFERbinom}}\label{sec:proof:ER}


For an inhomogeneous ER graph $G_n$ and its reconstruction $\hat{G}_n$, our goal is to determine the rate-distortion function,

\begin{equation}\label{eq:9}
    R(D) = \min_{P_{\hat{G}_n|G_n}: \mathbb{E}[d(G_n,\hat{G}_n)] \leq D } I(G_n;\hat{G}_n) 
\end{equation} 
where $\mathbb{E}[d(G_n,\hat{G}_n)]$ is the  expected Hamming distortion. To obtain this rate-distortion function of the ER graph, we first need to obtain a lower bound on the mutual information between $G_n$ and its reproduced version $\hat{G}_n$ and then show that this lower bound is achievable.

Consider a distribution $P_{\hat{G}_n|G_n}$ satisfying the distortion criteria. By setting $d_{i, j} \triangleq \min \{\mathbb{P}(E_{i, j} \oplus \hat{E}_{i, j}  = 1), \mathbb{P}(E_{i, j} \oplus \hat{E}_{i, j}  = 0)\}$, the distortion condition becomes  
\begin{align}
   D \geq \mathbb{E}[d(G_n,\hat{G}_n)]  =  \sum_{1 \leq i < j \leq n} \mathbb{E}[d(E_{i, j},\hat{E}_{i, j})]   = \sum_{1 \leq i < j \leq n} \mathbb{P}(E_{i, j} \oplus \hat{E}_{i, j} = 1) \geq \sum_{1 \leq i < j \leq n} d_{i, j}.\label{PerEdgeDistortion}
\end{align}
A lower bound on $R(D)$ can be obtained by evaluating the mutual information $I(G_n;\hat{G}_n)$ as follows

\begin{align}
I(G_n;\hat{G}_n) &= H(G_n) -H(G_n| \hat{G}_n)\nonumber\\
    &= \sum_{i<j} H(E_{i, j}) - \sum_{i<j} H(E_{i, j}|E_{1,2},E_{1,3}, \dots, E_{i-1,j},\hat{G}_n) \label{eq:11}\\
    &\geq \sum_{i<j} H(E_{i, j}) - \sum_{i<j} H(E_{i, j}|\hat{E}_{i, j}) \label{eq:12}\\
    & = \sum_{i<j} I(E_{i, j}; \hat{E}_{i, j})\nonumber \\
    & = \sum_{i<j} \max \{ 0, H(E_{i, j}) - H(E_{i, j} \oplus \hat{E}_{i, j}|\hat{E}_{i, j})\}  \label{mutualInfo}\\
    &\geq \sum_{i<j} \max \{ 0, H(E_{i, j}) - H(E_{i, j} \oplus \hat{E}_{i, j})\} \label{MutualERInh}\\
 &= \sum_{i<j} \max \{ 0, h_2(p_{i, j}) - h_2(d_{i, j})\}, \label{MutualERInh2}
\end{align}
where  (\ref{eq:11}) follows from the mutual independence of the $E_{i,j}$'s in an ER graph and the chain rule of the entropy, (\ref{mutualInfo}) follows from the fact that mutual information between any two random variables is non-negative, and (\ref{eq:12}) and (\ref{MutualERInh}) follow from the fact that conditioning reduces entropy. 

From (\ref{eq:9}), (\ref{PerEdgeDistortion}) and (\ref{MutualERInh2}), we clearly see that
\begin{align}
    R(D)  &=\min_{P_{\hat{G}_n|G_n}: \mathbb{E}[d(G_n,\hat{G}_n)] \leq D } I(G_n;\hat{G}_n) \\
    &\geq \min \left\{ \sum_{i<j} \max \{ 0, h_2(p_{i, j}) - h_2(d_{i, j})\}: \sum_{i < j} d_{i, j} \leq D \right\} \label{AchievabilityERInh}\\
     & = \sum_{i < j} h_2(p_{i, j}) - \max \left\{\sum_{i<j} \min \left\{ h_2(p_{i, j}), h_2(d_{i, j})\right\}\colon \sum_{i < j} d_{i, j} \leq D \right\}.\label{RateOptimisation}
\end{align}
We need to solve the optimisation problem
\begin{equation}
    \max \left\{\sum_{i<j} \min \left\{ h_2(p_{i, j}), h_2(d_{i, j})\right\}\colon \sum_{i < j} d_{i, j} \leq D \right\},\label{ObjectiveFunctionERinh}
\end{equation}
 where the $p_{i, j}$'s are fixed and the $d_{i, j}$'s are the optimization variables. We can easily argue that the optimization problem \eqref{ObjectiveFunctionERinh} is equivalent to the optimisation problem (in the standard form), 
\begin{equation}\label{eq:equiv:optimization}
\begin{aligned}
& \text{minimize}
& & -\sum_{i < j} h_2(d_{i, j}) \\
& \text{subject to}
& & \sum_{i < j} d_{i, j} = D, & i,j \in \mathcal{V}, \\
&&& d_{i, j} \leq \min \{p_{i, j}, 1 - p_{i, j}\}, & i,j \in \mathcal{V}.
\end{aligned}
\end{equation}
The equivalence can be shown by mapping an optimizer of \eqref{ObjectiveFunctionERinh} to a point in the constraint of \eqref{eq:equiv:optimization} such that the objective functions of both these problems evaluated at these points remain the same, and vice versa.

As in the proof of Theorem~\ref{RDFSBMcommunityknown}, we can use the Lagrangian function and the KKT conditions to derive the optimizer, which is
\begin{equation*}
        d_{i,j}^* = \min\left\{\min\{p_{i,j}, 1 - p_{i,j}\}, \mu \right\}, \end{equation*} 
where $\mu$ is chosen such that $\sum_{i,j} d_{i,j}^* \leq  D$. By plugging this solution into \eqref{RateOptimisation}, we get
\begin{align}\label{eq:lowerbound:ERinh}
    R(D)  \geq \sum_{i < j} \left[h_2(p_{i, j}) - h_2(d^{*}_{i, j})\right].
\end{align}
{\textbf{Achievability of the lower bound:}} We can show the achievability of the lower bound \eqref{eq:lowerbound:ERinh}  by finding a distribution $P_{\hat{G}_n|G_n}$ such that $ \mathbb{E}[d(G_n,\hat{G}_n)] \leq D$ and $I(G_n; \hat{G}_n)$ is the expression on the right-hand side of \eqref{eq:lowerbound:ERinh}. Similar to the SBM case, this achievability part of the proof can be completed by considering the conditional distribution 
\begin{align}
    P_{\{\hat{E}_{i,j}:{i<j}\}|\{E_{i,j}:{i<j}\}}= \prod_{i<j} \frac{P_{{E}_{i, j}|\hat{E}_{i, j}}P_{\hat{E}_{i, j}}}{P_{{E}_{i, j}}},
\end{align}
 where  $P_{{E}_{i, j}|\hat{E}_{i, j}}$ is a binary symmetric channel with the crossover probability $d^*_{i,j}$ and $P_{\hat{E}_{i, j}}(1)=\frac{p_{i,j} -d^*_{i,j}}{1 - 2d^*_{i,j}}$ and $P_{\hat{E}_{i, j}}(0)=\frac{1-p_{i,j}-d^*_{i,j}}{1 - 2d^*_{i,j}}$, which are non-negative by virtue of the way the $d^*_{i,j}$'s are defined.

\section{Conclusion and Future Work}
In this paper, we addressed the lossy compression of SBM by deriving its rate-distortion function under the Hamming distortion. Our approach involved considering the conditional rate-distortion function with community memberships as side information. In addition, we characterized the rate-distortion function of the Erd\H{o}s-R\'enyi random graph models. In both these models, the rate-distortion expressions contain a matrix $\mathbf{D}^{*}$, whose entries are chosen using a version of the ``reverse-water filling" solution.


In future work, we plan to explore the case where the distortion measure takes into account the community labels. This could lead to the problem of partial recovery in the stochastic block model \cite{Abbe2016GraphClusters}. We also aim to investigate distortion measures that effectively preserve the graph's structure during reconstruction, such as the measure considered in \cite{Bustin2022OnGraphs, Bustin2017OnMatrices}. These results could also be incorporated in applications such as graph anomaly detection through the use of coding algorithms and in reducing overhead during topology inference in networks.

\appendices

\makeatletter
\patchcmd{\@sect}{\@xsect}{\protected@edef\@currentlabelname{Appendix~\Alph{section}}\@xsect}{}{}
\makeatother

\ifCLASSOPTIONcaptionsoff
  \newpage
\fi



%

%


\bibliographystyle{ieeetr}

\bibliography{references_formatted}

\end{document}